\documentclass[12pt]{article}

\usepackage[margin=1in]{geometry}
\usepackage{amsmath, amssymb, mathtools}
\usepackage{graphicx}
\usepackage{microtype}
\usepackage{lmodern}
\usepackage[english]{babel}
\usepackage{csquotes}
\usepackage{enumitem}
\usepackage{setspace}
\usepackage[acronym]{glossaries}
\usepackage{hyperref}
\usepackage{amsthm}
\theoremstyle{definition}
\usepackage{mathtools}
\newtheorem{theorem}{Theorem}[subsection]
\newtheorem{proposition}[theorem]{Proposition}
\newtheorem{corollary}[theorem]{Corollary}

\newtheorem{postulate}[theorem]{Postulate}

\title{Rational Adversaries and the Maintenance of Fragility: A Game-Theoretic Theory of Rational Stagnation}
\author{Daisuke Hirota}

\begin{document}
\date{}
\maketitle
\begin{abstract}
Cooperative systems often remain in persistently suboptimal yet stable states. 
A suppressive adversary values unrealized social surplus, $u_D=U_{\text{ideal}}-U_{\text{actual}}$. In a Prisoner's Dilemma with preferences $u_i'=au_i+bu_j$, we characterize the ratios $w=b/a$ for which mutual cooperation and mutual defection are both equilibria. When the adversary can maintain access to this band, growth in the surplus at stake can make delaying collapse optimal. The adversary may also suspend intervention and return later. Active maintenance is worthwhile when preserving access to the fragile band raises the expected value of later intervention enough to cover its current cost. A two-state surplus process illustrates both paid maintenance and suspension followed by re-intervention. A reference-dependent extension bounds the change in action values under a shift of the benchmark.
Applications to social-media algorithms and political trust illustrate how adversarial rationality can deliberately preserve fragility.
\end{abstract}
\noindent\textbf{Keywords:} rational adversary; suppressive adversary; rational stagnation; fragile cooperation band; principle of potential loss; ratio of mutual recognition ($w$); cooperative surplus ($\Phi$); Bellman equation; adversarial intervention; meta-game; social media; institutional design.
\clearpage
\tableofcontents
\clearpage

\section{Introduction}
\label{sec:introduction}

\subsection{Problem Statement and Research Motivation}

A cooperative system may remain fragile even when an actor could precipitate its collapse. One possible explanation is that the actor expects a larger loss of cooperative surplus if collapse occurs later. The analysis asks when this incentive makes preserving fragility worthwhile.

The adversary values the difference between the welfare attainable under cooperation and the welfare realized at collapse. By choosing when to induce collapse, it can change the benchmark against which the realized outcome is valued.

\subsection{Central Research Question}

Classical game theory provides a robust framework for analyzing strategic interaction under fixed rule sets. However, this framework typically does not account for agents whose primary utility derives from degrading the cooperative potential of the system itself. This observation motivates the central research question:
\begin{quote}
\textbf{How does the introduction of a \emph{rational adversary}—whose payoff is defined by the society-wide \emph{potential loss}—transform the equilibrium and temporal dynamics of \emph{cooperative systems}?}
\end{quote}
The analysis connects equilibrium multiplicity in the players' stage game to the adversary's choice of when to induce collapse.

\subsection{Contributions and Positioning of This Paper}

This paper advances the study of strategic interaction on three levels.

\paragraph{First, a new formalization of adversarial rationality:}
The \textbf{suppressive adversary} values a shortfall in others' total welfare. Holding the cooperative benchmark fixed gives the same ranking of outcomes as a negative weight on that total. In the dynamic problem the adversary chooses when to realize the shortfall and therefore which date's benchmark enters its payoff.

\paragraph{Second, a new mechanism of sustained stagnation:}
The adversary can pay to preserve favorable terms for inducing collapse as the cooperative surplus grows. We call this choice \textbf{rational stagnation} when maintenance preserves the fragile cooperation band. The current cost and the discounted value of future intervention determine how long the adversary waits.

\paragraph{Third, foundations for a new research program:}
By inverting the standard mechanism-design problem, this paper lays the theoretical groundwork for \textbf{adversarial mechanism design}. Rather than maximizing social welfare, the question here is how an adversary would design the rules of the game to maximize future manipulability and fragility. This perspective offers a novel analytical lens for understanding how dysfunctional institutions, markets, and political systems can be intentionally designed and perpetuated.

\subsection{Methodological Approach and Scope of This Paper}

The \textbf{rational adversary} introduced here is not an empirical model that literally depicts the psychology of a particular individual or organization. Rather, it is positioned as an analytical metaphor—or ideal type—designed to extract and purify a specific \emph{logic} underlying complex phenomena.

The assumption of an \textbf{external adversary} in the model serves as an analytical device. It does not necessarily posit the existence of a single outside actor. Instead, it functions as a theoretical lens for analyzing dynamics oriented toward stagnation that emerge from incentive structures within the system and from the interactions among multiple actors. For example, a platform algorithm that maximizes engagement on social media may, absent any explicit hostility, come to embody the logic of the \textbf{suppressive adversary}—an instance of \emph{adversarial logic without hostility} in a \emph{cooperative system}. The aim of this paper is less to propose concrete defensive countermeasures than to offer a diagnosis that enables the structure of the problem to be understood from a new vantage point.

\subsection{Organization of This Paper}

The remainder of this paper is organized as follows. Section~\ref{sec:background} reviews the relevant literature, situates the three contributions advanced here within existing research, and identifies the conceptual gaps they address. Section~\ref{sec:analytical_results} presents a general Theorem concerning the adversary’s optimal strategy and analyzes its most salient implication, namely \textbf{rational stagnation}. Section~\ref{sec:discussion} applies the framework to two distinct types of cases: an endogenous adversary in social media and externally orchestrated polarization campaigns targeting online communities. Section~\ref{sec:conclusion} states the overall conclusions and limitations of this study, discusses implications for institutional design and AI ethics, and outlines directions for future work under the rubric of \textbf{adversarial mechanism design}.

\section{Literature Review and Conceptual Framework}
\label{sec:background}

\subsection{Introduction: Identifying the Conceptual Gap}

Classical game theory has compellingly demonstrated how rational choices under fixed rules can yield collectively irrational outcomes.\cite{Dawes1980,Hardin1968} However, the scope of that analysis typically presumes a closed system in which the rules of the game and the potential payoffs are exogenously given. This presumption creates a critical blind spot for analyzing real-world \emph{cooperative systems}: namely, the existence of agents who profit not from winning the game, but from keeping the game itself inefficient—deriving gains precisely from the failure of cooperation.

The following comparisons locate the source of the adversary's payoff and the date at which it is received. Both features matter for distinguishing suppression from the protection of a private income stream.

\subsection{A New Taxonomy of Adversarial Rationality}

The following three motives organize the comparison. A single institution may face more than one of them.

\subsubsection{Value-Extractive Adversaries}

Adversaries in this category derive utility from directly acquiring or destroying specific assets or states of the system. A representative example is the family of \textbf{security games} developed in the contexts of infrastructure defense and cybersecurity.\cite{tambe2011security} In these models, the attacker’s utility typically arises from \textbf{direct gains} such as asset destruction or capture, and is defined independently of the system’s \textbf{latent capacity}. For instance, in \textbf{Stackelberg Security Games (SSG)}, the adversary’s (attacker’s) utility is tied to the \textbf{value} of the target asset (e.g., success rewards and capture penalties).\cite{paruchuri2008playing,Korzhyk2011} Likewise, in the \textbf{Byzantine Agreement (BA)} problem, classical results establish the limits of safety and resilience that consensus protocols can satisfy, rather than specifying an explicit utility function for the adversary.\cite{Lamport1982} In essence, these adversaries seek to \textbf{extract} concrete “value” from the system as destroyers.

\subsubsection{Evasive Adversaries}

Adversaries in this category seek to achieve a separate primary goal while \emph{evading} detection or defense mechanisms. A canonical example is \textbf{adversarial machine learning}, particularly \textbf{evasion attacks} in spam filtering and related applications.\cite{Barreno2010} Here, the adversary’s primary objective (e.g., delivering spam) motivates minimal modifications to inputs that circumvent the classifier, albeit at an explicit \textbf{cost}.\cite{Dalvi2004,Lowd2005} Their utility can thus be formalized as a \textbf{trade-off} between achievement of the primary objective (e.g., misclassification or delivery success) and the costs of evasion (e.g., the magnitude of feature changes or the number of manipulable steps).\cite{Biggio2013} Moreover, an \emph{evasive} subtype that tempers attack intensity—because the disutility of being detected is comparatively large—can be positioned as a \emph{cautious (risk-averse)} attacker within this class.\cite{Barreno2010}

\subsubsection{Suppressive Adversaries}

By contrast, this paper introduces a third category: the \textbf{suppressive adversary}. The utility of this adversary is defined by what this paper terms the \textbf{principle of potential loss}, $u_{D}=U_{ideal}-U_{actual}$. Such an adversary need not steal or destroy anything from the system. Rather, utility arises from the \textbf{gap between the ideal and the actual state}—that is, the shortfall from what the system would have achieved under optimal functioning. This captures a fundamental distinction between adversaries who aim to win the game and adversaries who aim to preserve a state in which no one can win. In this sense, the adversary is a \textbf{governor} who holds the system’s potential hostage. Hereafter, for convenience, this paper refers to the suppressive adversary simply as the “adversary” or the “\textbf{rational adversary}.”

\subsection{Conceptual Distinctions from Other Nonstandard Utility Functions}

To compare these preferences, fix the outcomes being ranked, the agent receiving the payoff, and the benchmark at the time of choice.

\subsubsection{Difference from Regret Theory: Locus of Utility}

Regret theory compares the decision maker's realized outcome with the outcome the same person would have obtained from a different choice in the same state.\cite{Loomes1982} Here the shortfall concerns the welfare of other players and enters the adversary's payoff.

The benchmark is the aggregate payoff from mutual cooperation at that date. Identifying an attainable benchmark is part of specifying an application.

\subsubsection{Difference from Models of Spite and Envy: Object of Utility}

Levine models social preferences by adding weighted payoffs of other players to the agent's direct utility.\cite{Levine1998} When the weights on the two other players are equal and negative and the agent's direct payoff is held fixed, the agent ranks stage outcomes by decreasing total welfare. The suppressive adversary gives the same ranking at a fixed benchmark. At two possible collapse dates with the same realized welfare and different cooperative benchmarks, its gross stopping rewards differ by the change in the benchmark.

By contrast, \textbf{envy} (inequity aversion) models provide a more refined formalization of aversion to \textbf{relative allocations}.\cite{Fehr1999}
\[
U_i = u_i - \alpha_i \max\{u_j - u_i, 0\} - \beta_i \max\{u_i - u_j, 0\}.
\]
In this specification, agent $i$’s utility decreases with both disadvantageous inequality ($u_j - u_i > 0$, corresponding to envy) and advantageous inequality ($u_i - u_j > 0$, corresponding to guilt). The ERC model likewise incorporates an agent’s \textbf{relative rank} (share) in total payoffs into utility.\cite{Bolton2000} Relatedly, work on the “Joy of Destruction” shows that the act of destroying others’ resources can itself yield positive utility.\cite{abbink2009pleasure}

The \textbf{suppressive adversary}, by contrast, is wholly \emph{indifferent} to \textbf{relative payoffs among players}. Her utility depends only on the \textbf{absolute loss of social surplus} relative to what the system ought to produce under ideal functioning—namely, the following \textbf{potential differential} in a two-player (dyadic) system:
\[
    {\text{potential differential}} = 2R - (u_A + u_B) = U_{ideal}-U_{actual}
\]
Here, $R$ is the benchmark payoff each player would obtain under mutual cooperation. Whether $u_A > u_B$ or $u_B > u_A$ is immaterial to the adversary; The adversary prefers that the sum $(u_A + u_B)$ be minimized.

It follows that the focus of analysis in this paper shifts away from \textbf{interpersonal competition and relative standing} toward \textbf{strategies that intentionally maintain or amplify system-wide dysfunction}.

\subsection{Comparison Between Rational Stagnation and Existing Mechanisms of Stagnation}

The \textbf{rational stagnation} derived in this paper occupies a distinct position among existing theoretical accounts that explain the persistence of suboptimal equilibria.

\subsubsection{Difference from Vested-Interests Models: Source of the Payoff}

Public choice theory explains the persistence of inefficient institutions by the presence of \emph{vested interests}.\cite{Acemoglu2000} Vested actors resist reforms that threaten the flow of benefits they derive from the status quo.

Both actors can anticipate future benefits. A vested interest values its own rents or political power after reform. The suppressive adversary values the shortfall in others' aggregate welfare. Holding its own rents and intervention costs fixed, an additional unit of unrealized surplus at collapse raises its gross payoff by one unit. Persistence of the status quo alone cannot identify this motive.

In the dynamic model the adversary collects the shortfall once, on inducing collapse. A period of maintenance costs $C_m(t)$ and preserves the opportunity to collect a larger shortfall later. An increase in the cooperative benchmark can therefore make waiting more attractive even while the adversary receives no current rent.

\subsubsection{Difference from Poverty-Trap Models: Nature of the Mechanism}

A poverty trap denotes a structural mechanism by which poverty as an initial condition is self-reinforcing and reproduced over time.\cite{Banerjee1993, dasgupta1986inequality} In game-theoretic terms, it can be described as a failure of cooperation whereby each agent, taking others’ investment levels as given, lacks the incentive to unilaterally move to the high equilibrium, so that the aggregate settles at a Pareto-inferior low equilibrium.\cite{cooper1988coordinating, azariadis1990threshold, murphy1989industrialization} This form of stagnation is an \textbf{autonomously emergent equilibrium} arising from decentralized interactions among actors; it does not posit intentional control by a specific external actor.

Here the external actor chooses whether to maintain access to a fragile state. An increase in the surplus that cooperation could generate can strengthen the incentive to preserve that access.

\subsection{Conclusion: The Redefined Conceptual Gap}

On the basis of the foregoing analysis, the conceptual lacuna addressed by this paper can be redefined as follows. Prior research has theorized primarily: (1) value-extractive adversaries targeting specific assets (security games); (2) internal players focused on relative payoffs or affective rewards (models of spite, envy, and regret); (3) conservative strategies aimed at preserving realized benefits (vested-interests models); and (4) emergent stagnation driven by bounded rationality or informational constraints (poverty traps).

The model combines a payoff measured against contemporaneous cooperative welfare with control over the date at which a loss is realized. This combination gives a reason to pay for continued access to a fragile state.

\section{Analytical Results}
\label{sec:analytical_results}

\subsection{Model Formulation}
\label{sec:model_formulation_jp}

Consider a Prisoner's Dilemma with players $A$ and $B$ and $T>R>P>S$. Mutual cooperation yields the benchmark $U_{\text{ideal}}=2R$. It maximizes total stage payoff exactly when $2R\ge T+S$. We impose this condition when interpreting $2R$ as the welfare maximum. The unique Nash equilibrium of the original game is mutual defection, with total payoff $2P$.

Introduce an external \textbf{rational adversary} $D$. The adversary influences the game by manipulating the players’ \emph{subjective utility} functions. For each player $i \in \{A,B\}$, the transformed utility is given by
\begin{equation} \label{eq:modified_utility_jp}
u'_i = a u_i + b u_j, \quad (i \neq j, a>0, b \ge 0)
\end{equation}
where $u_i$ is player $i$'s objective payoff. Let $w=b/a$. When the thresholds in Appendix~\ref{sec:appendix_equilibrium_jp} satisfy $w_{\min}\le w_{\max}$, the transformed game has the following regions.
\begin{itemize}
    \item $w < w_{\min}$: the unique equilibrium is $(D,D)$ (distrust phase).
    \item $w_{\min} \le w \le w_{\max}$: both $(C,C)$ and $(D,D)$ are Nash equilibria (the \textbf{fragile cooperation band}).
    \item $w > w_{\max}$: the unique equilibrium is $(C,C)$ (cooperation phase).
\end{itemize}

The adversary’s utility is defined by the \textbf{principle of potential loss} and is tied to the players’ objective outcomes:
\begin{equation} \label{eq:adversary_utility_jp}
u_D = U_{\text{ideal}} - U_{\text{actual}} = 2R - (u_A + u_B)
\end{equation}
Inside the fragile band, a player who expects the other to defect is willing to defect too. Above the band, the same expectation induces cooperation. Thus mutual cooperation today can coexist with a relatively inexpensive option to induce mutual defection later. Appendix~\ref{sec:appendix_thresholds_jp} derives this cost advantage under an explicit rule for coordinating play at collapse.

\paragraph{Remark (Emergence of Asymmetric Equilibria)}
The thresholds are $w_{\min}=(T-R)/(R-S)$ and $w_{\max}=(P-S)/(T-P)$. The coexistence interval has positive width exactly when $w_{\min}<w_{\max}$. If their order is reversed, $w_{\max}<w<w_{\min}$ gives two asymmetric pure equilibria and one symmetric mixed equilibrium. Below both thresholds defection strictly dominates, and above both cooperation strictly dominates. At a threshold, a best response can be weak.

\subsection{General Theory of Adversarial Strategy}
\label{sec:general_theory_jp}

Let $\Phi_t=2(R_t-P)>0$ be the gross loss realized if the adversary induces collapse at date $t$. Collapse pays once. Maintenance costs $C_m(t)$ now and keeps the system in the fragile band. Suspension lets the system evolve without active control for one period, with zero current payoff to the adversary. At the next date the adversary can choose again. We compare the actions at a state where maintenance is feasible, the action values are finite, and the optimal value over maintenance controls is attained (as with a finite control set). Its transition law includes the consequences of the chosen control. The players' stage-game equilibria characterize the states in which either mutual cooperation or mutual defection can be sustained. Selecting between these outcomes requires a rule for play.

\begin{theorem}[Optimal intervention at a given state]
\label{thm:general_dichotomy_jp}
Let $S_t=\Phi_t-C_c(t)$ be the value of inducing collapse, $M_t=\delta E_t^m[V_{t+1}]-C_m(t)$ the value of maintenance, and $A_t=\delta E_t^0[V_{t+1}]$ the value of suspension. The superscripts specify the transition laws under maintenance and inaction. For a single feasible maintenance control, maintenance is optimal exactly when
\[
 M_t\ge S_t
 \quad\text{and}\quad
 \delta\big(E_t^m[V_{t+1}]-E_t^0[V_{t+1}]\big)\ge C_m(t).
\]
With several controls compare their largest action value with $S_t$ and $A_t$. Suspension is optimal when $A_t\ge\max\{S_t,M_t\}$, and collapse when $S_t\ge\max\{M_t,A_t\}$. All maximizing actions are optimal in a tie.

\end{theorem}
\begin{proof}
Condition on the first action. Collapse pays $S_t$ once. Both maintenance and suspension lead to the next decision date under their respective transition laws. Subtracting $A_t$ from $M_t$ gives the second inequality. Maximizing over the available first actions gives the remaining comparisons.

\end{proof}

A higher expected payoff from future collapse can make maintenance strictly optimal even at a positive maintenance cost. Appendix~\ref{sec:appendix_thresholds_jp} gives an example in which the adversary maintains the band at the low-surplus state and induces collapse on reaching the high-surplus state.

\subsection{Analysis of Rational Stagnation under Ideal Conditions}
\label{sec:ideal_case_jp}

To clarify the implications of Theorem~\ref{thm:general_dichotomy_jp}, consider an ideal environment in which intervention costs are zero and cooperation invariably engenders learning.

\begin{postulate}[Postulate of Cooperative Learning]
\label{post:learning}
Assume a socio-economic environment in which repeated cooperation promotes learning and efficiency, yielding a nondecreasing cooperative payoff $R_{t+1} \ge R_t$.
\end{postulate}

Under these ideal conditions, Theorem~\ref{thm:general_dichotomy_jp} yields the following sharper conclusion.

\begin{corollary}[Sufficient Condition for Delaying Collapse]
\label{cor:stagnation}
Suppose intervention is costless, maintenance followed by collapse next period is feasible, and $\Phi_t>0$. Write $g_t=E_t^m[\Phi_{t+1}]/\Phi_t-1$ under that maintenance action. Immediate collapse is strictly suboptimal at date $t$ if
\begin{equation} \label{eq:stopping_condition_jp}
\delta > \frac{1}{1+g_t}.
\end{equation}
\end{corollary}
\begin{proof}
The feasible plan pays $\delta E_t^m[\Phi_{t+1}]$, which exceeds $\Phi_t$ under \eqref{eq:stopping_condition_jp}. Further opportunities to wait can only raise the value of continuation. To conclude that active maintenance is optimal one must also compare $M_t$ with $A_t$, since suspension likewise permits later intervention.
\end{proof}

Postulate~\ref{post:learning} allows constant $R_t$. On such a path, a delay of $k>0$ periods yields $\delta^k\Phi_t<\Phi_t$ at zero cost. The sufficient condition compares the surplus available at two dates.

\section{Discussion: Applications and Interpretations}
\label{sec:discussion}

\subsection{Introduction}

Theorem~\ref{thm:general_dichotomy_jp} compares the value of inducing collapse with the value of retaining access to a fragile system. The third action, suspension, permits later intervention. Its value depends on how access to the fragile band evolves while the actor is inactive.

This section applies the general theory to social phenomena and assesses its explanatory power and observational implications. The focus here is the \textbf{empirical interpretation} of the adversary’s utility function
\[
u_D = U_{ideal} - U_{actual}
\]
defined in the previous section. An empirical shortfall requires an attainable benchmark and a realized outcome expressed on the same scale, with higher values indicating greater welfare. The theoretical benchmark $2R_t$ can change over time.

The interval $[w_{\min},w_{\max}]$ describes equilibrium multiplicity. Translating this interval into an observed range of trust or exit rates requires a measurement model and a rule for equilibrium selection.

Building on this observational framework, this section presents two case studies.\\
\textbf{Case 1: Adversarial rationality in social media.} A platform may earn engagement revenue while trust remains low. The following illustration asks when the resulting action ranking can resemble that of an actor paid for a terminal welfare shortfall.

\textbf{Case 2: Meta-game intervention in political and institutional systems.} In this context, societal trust in the state or institutions constitutes the ideal ($U_{\text{ideal}}$), while realized trust scores and political-stability indicators constitute the actual ($U_{\text{actual}}$). External players (states, media, or other strategic actors) intervene in the trust architecture of rival polities to maximize their own payoffs. The aim is not necessarily \emph{immediate destruction}, but rather to keep the rival’s growth potential and trust level \emph{within a manageable range} so as to secure competitive advantage. In this way, the \emph{meta-game} appears as a hierarchical structure in which multiple layers of rational players influence one another.

The two applications give possible interpretations of the adversary's incentives. Distinguishing them empirically requires evidence on the source and timing of the actor's payoff.

\subsection{Definitions of the Ideal and the Actual, and an Observational Framework for Intervention Variables}
For a proposed welfare indicator $x_t$, the observed outcome is $U_{\text{actual}}=x_t$ and the shortfall is $b_t-x_t$, where $b_t$ is a benchmark on the same scale. A normative target, a survey ceiling, and the welfare attainable under mutual cooperation have different interpretations. Using a survey ceiling as $2R_t$ requires evidence that it measures attainable cooperative welfare. Estimating the effect of intervention also requires a comparison with how the outcome would evolve under suspension.

\subsection{Ambiguous Tipping Points and the Probabilistic Structure of Rational Stagnation}
\label{sec:probabilistic_tipping}

An empirical tipping region could be estimated from repeated measurements of behavior and trust. A gradual change in exit rates alone cannot establish that multiple equilibria coexist.

If data and a behavioral model permit estimation of the two theoretical thresholds, statistical uncertainty can be reported as intervals
\[
[w_{\min}^-, w_{\min}^+] \quad \text{and} \quad [w_{\max}^-, w_{\max}^+] .
\]
These are uncertainty intervals for the endpoints of the theoretical band. A probability of switching between equilibria would require a separate behavioral model.

Slow movements in trust or exit rates may be consistent with continued maintenance. To infer that an actor is choosing maintenance, one would also need to observe its interventions and estimate how the system would evolve during a pause.

\subsection{Case Study 1: Platform Algorithms as an Endogenous Adversary}

\subsubsection{Modeling Assumptions}
We model a social-media platform as a \emph{meta-player} in a repeated game played by its users.
\begin{enumerate}
    \item \textbf{Players and cooperative payoff:} The players are platform users. When they share truthful, high-quality information and engage in constructive deliberation, a collective \emph{cooperative payoff} $R_t$ (shared understanding, networks of trust, and social capital) is generated.
    \item \textbf{Adversary and utility function:} The \emph{adversary} is the platform’s algorithmic architecture itself. Its utility (i.e., platform revenue) is not a direct function of $R_t$, but is maximized by user \emph{engagement} (time-on-platform, interactions).
    \item \textbf{Engagement maximization:} Engagement is assumed to be maximized in environments where disagreement persistently arises.
    \begin{itemize}
        \item \textbf{Cooperative/stable state ($w > w_{\max}$):} A state dominated by stable agreement and trust exhibits high predictability and weak incentives for users to post rebuttals, corrections, or identity-affirming content. This implies lower engagement.
        \item \textbf{Fragile cooperation band ($w_{\min} \le w \le w_{\max}$):} Suppose behavior in this band sustains enough disagreement to increase engagement. The stage-game thresholds alone determine neither engagement nor the frequency of conflict.
    \end{itemize}
    \item \textbf{Intervention costs and risk:} The adversary’s \emph{intervention cost} $C_m(t)$ is understood as a composite \emph{risk-management} cost required to maintain fragility. Excessive polarization ($w < w_{\min}$) triggers firestorms, the spread of hate speech, and public controversies, risking mass user exit and regulatory sanctions. Expenditures on content moderation, public relations, and lobbying are included in these costs.
\end{enumerate}

\subsubsection{Ideal and Actual in the SNS Model}
In social media, the \textbf{ideal} ($U_{\text{ideal}}$) is the \emph{maintenance of a healthy and trustworthy deliberative space}. Operationally, this can be defined as attaining the maximum values on user-satisfaction or interpersonal-trust indices. The ideal represents the best social state the platform could achieve—stable cooperative networks, low misinformation, and a healthy dialogic structure.

A satisfaction or trust index can measure the realized outcome on the same scale as its cooperative benchmark. Counts of misinformation or conflict may help explain variation in that index. Subtracting a raw count from a trust score would mix units and reverse the welfare orientation of an adverse indicator. Coexistence of two stage equilibria specifies which outcomes are available. An empirical claim about the frequency of conflict requires evidence on play within the band.

\subsubsection{Equilibrium of Rational Stagnation}
A platform receiving engagement revenue each period has a flow payoff. Its maintenance value includes current revenue net of moderation costs and the discounted value of retaining its users. If this exceeds the values of allowing collapse and of suspending intervention, the platform prefers maintenance. This ranking can coincide with that of the suppressive adversary even though the source and timing of revenue differ. Observing sustained engagement alone does not identify a preference for unrealized social surplus.

\subsection{Case Study 2: Meta-game Intervention by an Exogenous Adversary}

\subsubsection{Modeling Assumptions: A Hierarchical Game}
We now treat the situation in Case Study~1 as a subgame and introduce external players who play a higher-level \emph{meta-game}.
\begin{enumerate}
    \item \textbf{Structure of the meta-game:} At the upper layer, players such as states and political organizations compete in a meta-game where payoffs consist of geopolitical influence and advantages in domestic politics.
    \item \textbf{Intervention into the subgame:} Players in the meta-game (the \emph{exogenous adversaries}) intervene in lower-layer games (specific online communities) as a tactic to achieve their own objectives.
    \item \textbf{Cooperative surplus and the adversary’s utility:} The cooperative surplus $R_t$ of the target community (trust, cohesion, and collective problem-solving capacity) constitutes part of the rival’s state capacity and support base in the meta-game. Accordingly, the exogenous adversary’s utility is defined by the community’s \emph{potential loss}, $U_{ideal} - U_{actual}$. Community dysfunction (a decrease in $U_{\text{actual}}$) weakens the rival and thus yields direct gains to the adversary in the meta-game.
    \item \textbf{Low-cost attacks:} By exploiting the Internet’s \emph{hyper-public} character and the fact that target communities are \emph{open yet clustered} due to language and platform segmentation, adversaries can operate at very low intervention cost $C_m(t)$. Signal amplification by bots and the injection of misinformation are cheaper and less risky than physical covert operations.
\end{enumerate}

\subsubsection{Ideal and Actual in the Meta-game Model}
Let $x_t$ be the share of respondents reporting trust. A stipulated ceiling of one gives the descriptive shortfall $1-x_t$. Interpreting it as foregone attainable trust requires an estimate of trust under cooperation. Coverage, exposure and survey responses could be measured together. Their correlation alone does not establish that a feasible intervention caused a change in trust or in the recognition ratio $w$.

\subsubsection{Equilibrium of Rational Stagnation}
In the terminal-loss specification, suppose intervention can induce collapse now or maintain the band while the surplus at stake changes. The adversary actively maintains the band if maintenance has greater value than both collapse and temporary suspension. A meta-game that pays for a rival's weakness each period would add that flow to the maintenance payoff. Its stopping boundary must then be computed with the added flow.

\subsubsection{Cultivating Fragility: The Rational Adversary’s Paradoxical Support}
The theory yields a more refined strategic implication. The adversary’s goal is not merely to produce stagnation. To maximize $u_D = U_{ideal} - U_{actual}$, it can be \emph{desirable} that the system’s \emph{cooperative surplus} $R_t$ \textbf{grows}. The higher the expected growth rate of $R_t$, the larger the future harvestable \emph{potential loss} in expectation.

An adversary may thus prefer a target whose potential is growing. Deliberate investment in that growth would require a further control over $R_t$, its cost, and its effect on the surplus. In the model of Appendix~\ref{sec:appendix_dynamic_general_jp}, the growth process is exogenous.

\section{Conclusion and Future Directions}
\label{sec:conclusion}

\subsection{Summary of Theoretical Contributions}
The subjective payoff transformation identifies when both symmetric outcomes are equilibria of the stage game. Where the adversary can preserve access to that band, maintenance is optimal when the value of future intervention covers the cost of waiting and exceeds the available alternatives. In the example of Appendix~\ref{sec:appendix_thresholds_jp}, the adversary pays to wait for a larger unrealized surplus. The source of this gain distinguishes the mechanism from protection of a private rent.

\subsection{Implications for AI and Institutional Design}
This theory carries the following implications for AI alignment and institutional design.

\textbf{Implications for AI systems:}
AI systems designed to optimize proxy metrics in complex social environments (e.g., user engagement) can \emph{endogenously} develop the strategic structure analyzed here for the \textbf{rational adversary}. That is, even without explicit hostility, the \emph{maintenance of fragility} within a controllable range may become optimal relative to the objective function. This perspective reframes phenomena such as “specification gaming” and “reward hacking” in AI alignment not merely as design flaws but as the \emph{rational} consequences of optimization.

\textbf{Implications for institutional design:}
Raising the cooperative surplus can benefit the players and increase the loss available to an adversary at collapse. An institutional intervention can also change the adversary's access to the fragile band. Evaluating the two effects together requires a planner's objective and a feasible set of institutional changes.

\subsection{Theoretical Extensions: Structural Adversaries and Adversarial Mechanism Design}
This paper analyzed the rationality of a \emph{psychological adversary} who operates within an existing rule system by manipulating beliefs. A key direction for future theory is to treat the rationality of a \textbf{structural adversary}—an agent capable of manipulating the rules themselves.

The structural adversary’s objective is to modify the system-wide payoff structure $(T, R, P, S)$ in ways that \emph{institutionally extend} the effective domain of psychological manipulation. The present framework offers a foundation for understanding why such structural interventions can be rational. By quantifying the tactical value of belief manipulation, it opens a path to analyze institutional design itself as an object of rational intervention.

This agenda constitutes a research program that may aptly be called \textbf{adversarial mechanism design}: a theoretical framework that describes how a rational agent—whose objective is not the maximization of social welfare but the maximization of future \emph{manipulability and fragility}—can design or revise institutions. This paper provides the foundational theoretical structure for that program.

\subsection{Limitations and Assumptions of the Theory}
\label{sec:limitations}
The model of \textbf{rational stagnation} presented in this paper offers a unified representation of intervention dynamics in adversarial games through the differential utility between the ideal and the actual. While the framework is extendable to diverse settings, its scope is bounded by the following key limitations and assumptions.

\subsection*{(1) Definition of the Ideal and Normative Premises}
The ideal utility $U_{\text{ideal}}$ is a normatively determined upper bound that depends on institutions and culture, and its specification may involve an element of arbitrariness. Although this paper mitigates arbitrariness by using observed upper bounds, the assumption that the ideal reflects ethical or social “good” cannot be entirely eliminated. See also the reference-dependent generalization in Appendix~\ref{appendix:math_jp}, \S\ref{sec:appendix_diffpayoff}, where shift-stability is established (Proposition~\ref{prop:ref-shift}).

\subsection*{(2) Directionality of the Differential and Identification}
Holding costs fixed, the adversary prefers a larger shortfall at collapse. Independent measures of the benchmark and realized welfare are needed because both can move together. Proposition~\ref{prop:ref-shift} gives a sufficient margin for an action ranking to survive a change of reference.

\subsection*{(3) Time Lags and Dynamic Stability}
Changes in $U_{actual}(t)$ exhibit inertia and delay, so instantaneous differentials alone may be insufficient to explain stability. To address this, the analysis introduces a generalized \emph{Bellman equation} that uses the discount factor $\delta$ and the maintenance-cost structure $C_m(t)$. Nonetheless, empirical identification of these parameters remains limited. For mass simultaneous interventions, Appendix~\ref{sec:appendix_mass_intervention} gives a one-line macro dynamics with a local stability test based on the slope $J_t$ at a differentiable fixed point (Proposition~\ref{prop:local-stability}).

\subsection*{(4) Bounded Rationality and Nonlinear Preferences}
Nonlinear sensitivities among players are captured endogenously by a function $F(w)$. Under the assumption $F'(w) \ge 0$, the model can reproduce diminishing empathy and tipping effects without explicitly imposing bounded rationality. However, the functional form of $F(w)$ is context-dependent, and identifying a general form empirically is nontrivial. Complementarily, the difference-based payoff in Appendix~\ref{sec:appendix_diffpayoff} accommodates asymmetric nonlinear sensitivities via $(\beta_i^\pm,\gamma_i^\pm)$ with preservation of a strict action ranking under the margin condition of Proposition~\ref{prop:ref-shift}.

\subsection*{(5) Testability and Domain of Application}
For empirical validation, each domain requires clear definitions and calibration of:
\begin{itemize}
    \item the observational indicators of the ideal utility $U_{\text{ideal}}$
    \item the realized measurements of the actual utility $U_{\text{actual}}$
    \item proxy variables for the \textbf{ratio of mutual recognition} $w=b/a$
    \item (for mass interventions) calibration of $(\beta^\pm,\gamma^\pm)$, sensitivity $\eta$, diffusion gain $\kappa$, and damping $\rho$ (Appendix~\ref{sec:appendix_mass_intervention})
\end{itemize}
Testing the maintenance mechanism also requires predictions for play and welfare under passive evolution and information about the costs of both interventions.

\vspace{1em}
In sum, while grounded in the simple structure of an \emph{ideal–actual} differential, the framework is general enough to incorporate dynamics, nonlinearity, multiple agents, and bounded rationality. Applications should, however, proceed with awareness of the normative definition of the ideal and the limits of empirical estimation. Accordingly, this paper proposes not a universal theory but a \emph{general model as a minimal, empirically testable structure}.

\subsection{Conclusion}
The model identifies a reason to preserve access to a fragile system. A larger future cooperative surplus can justify paying a cost today to retain the option of inducing collapse. How long this incentive persists depends on the surplus process, intervention costs, and the alternatives available to the actor.

By formalizing this theory, phenomena previously labeled as “irrational failures” can be analyzed as \emph{rational strategies}. In any social or epistemic system, the ultimate task is not merely to sustain cooperation, but to establish \emph{institutional designs that protect the processes of learning and cooperation from a rational adversary}.

\clearpage
\appendix
\section{Mathematical Appendix}
\label{appendix:math_jp}

\subsection{Formal Definition of the Adversarial Game}
\label{sec:appendix_definitions_jp}

The model has ordinary players $A$ and $B$ and an adversary $D$. The ordinary players choose $C$ or $D$ in a Prisoner's Dilemma with objective payoffs $T>R>P>S$. Their stage-game best responses use the adjusted utilities below. Appendix~\ref{sec:appendix_dynamic_general_jp} specifies the adversary's decision problem.

Each player’s \emph{subjective utility} $u'_i$ is altered by the adversary via psychological coefficients $(a,b)$:
\begin{equation} \label{eq:appendix_utility_jp}
    u'_i = a u_i + b u_j, \quad (i \neq j, a>0, b \ge 0).
\end{equation}

The adversary’s utility $u_D$ is defined by the \textbf{principle of potential loss} and depends on the players’ objective outcomes:
\begin{equation} \label{eq:appendix_adversary_utility_jp}
    u_D = U_{\text{ideal}} - U_{\text{actual}} = 2R - (u_A + u_B).
\end{equation}

\subsection{Derivation of Equilibrium Conditions and the Fragile Cooperation Band}
\label{sec:appendix_equilibrium_jp}

Under the transformed utilities in \eqref{eq:appendix_utility_jp}, Nash equilibrium conditions follow from best-response analysis:
\begin{itemize}
    \item $(C,C)$ is a Nash equilibrium $\iff a(R-T)+b(R-S) \ge 0 \iff w \ge \frac{T-R}{R-S} =: w_{\min}$,
    \item $(D,D)$ is a Nash equilibrium $\iff a(P-S)-b(T-P) \ge 0 \iff w \le \frac{P-S}{T-P} =: w_{\max}$,
\end{itemize}
where $w=b/a$ denotes the \textbf{ratio of mutual recognition}. A necessary and sufficient condition for the existence of the \textbf{dual equilibrium band} (i.e., the interval on which both $(C,C)$ and $(D,D)$ are equilibria, coinciding with the \emph{fragile cooperation band}) is $w_{\min} \le w_{\max}$, which is equivalent to $(T-R)(T-P) \le (P-S)(R-S)$. The difference of the two cross-products factors as $(T-S)(T+S-R-P)$. Thus the band is nonempty exactly when $T+S\le R+P$, with positive width under strict inequality.

\subsection{Static Analysis (One-Shot Game)}
\label{sec:appendix_static_jp}

Across all action profiles the minimum objective total is $\min\{2P,T+S\}$. When $2P\le T+S$, the adversary can attain the minimum $2P$ with a common weight $w=0$, which makes defection strictly dominant. We use this payoff ordering in the baseline. If $T+S<2P$, the profile minimum occurs at $(C,D)$ and $(D,C)$. Either profile would require $w_{\max}\le w\le w_{\min}$ to be a Nash equilibrium. Yet $T+S<2P<R+P$ implies $w_{\min}<w_{\max}$. The asymmetric profile minimum is thus unavailable as a pure equilibrium under a common weight.

\subsection{General Dynamic Model: Intervention Costs and Stochastic Payoffs}
\label{sec:appendix_dynamic_general_jp}

We construct a more general dynamic model that relaxes the \emph{Postulate of Cooperative Learning} and allows for intervention costs. Let the discount factor satisfy $0 < \delta < 1$.

\paragraph{Constituent elements of the model}
\begin{itemize}
    \item \textbf{Stochastic cooperative payoff:} The cooperative payoff $R_t$ at time $t$ is assumed to fluctuate probabilistically in the next period. Let $E_t[R_{t+1}]$ denote the expectation of $R_{t+1}$ at time $t$.
    \item \textbf{Intervention costs:} The adversary incurs costs when taking actions at time $t$:
    \begin{itemize}
        \item $C_c(t)$: cost of \emph{collapse induction} (e.g., expenditures on a disinformation campaign),
        \item $C_m(t)$: cost of \emph{maintenance of fragility} (e.g., ongoing monitoring and interventions required to sustain social division).
    \end{itemize}
\end{itemize}

\paragraph{Supplement to the derivation of Corollary~\ref{cor:stagnation}}
Let $\Phi_t=2(R_t-P)>0$ and $g_t=E_t^m[\Phi_{t+1}]/\Phi_t-1$ under a feasible maintenance action. At zero intervention costs the one-period comparison gives
\[
M_t>S_t\;\Longleftarrow\;\delta\,\mathbb{E}_t^m[\Phi_{t+1}]>\Phi_t
\;\Longleftrightarrow\;\delta\,(1+g_t)\,\Phi_t>\Phi_t
\;\Longleftrightarrow\;\delta>\frac{1}{1+g_t}.
\]
This is a sufficient condition; even under dynamic optimization (with $V_{t+1}\ge \Phi_{t+1}$), the same inequality supports a preference for continuation.

\paragraph{Generalized Bellman Equation}
Write $E_t^m$ and $E_t^0$ for expectations conditional on the current state under maintenance and passive evolution, respectively. For a single maintenance control, the discounted value is
\begin{equation} \label{eq:bellman_general_jp}
 V_t=\max\big\{\underbrace{\Phi_t-C_c(t)}_{S_t},\
 \underbrace{\delta E_t^m[V_{t+1}]-C_m(t)}_{M_t},\
 \underbrace{\delta E_t^0[V_{t+1}]}_{A_t}\big\}.
\end{equation}
Suspension incurs no intervention cost and leaves future actions available at their ordinary state-dependent costs. An additional restart charge could be modeled by recording whether the previous period was inactive and charging only on the first subsequent intervention.

\noindent\textit{Reference-dependent extension.} 
A unified, difference-based payoff (including positive/negative sensitivity and mass simultaneous interventions) 
is developed in Appendix~\ref{appendix:math_jp}, 
\S\ref{sec:appendix_diffpayoff}–\S\ref{sec:appendix_mass_intervention}.

\paragraph{Assumption of an Absorbing Post-Collapse State}
If the adversary stops (induces collapse) at time $t$, the system transitions to an absorbing state $\mathsf{Collapse}$ in which $A$ and $B$ subsequently repeat $(D,D)$, fixing the objective total payoff permanently at $2P$. The adversary’s realized payoff at time $t$ is therefore
\[
u_D^{\text{stop}}(t) \;=\; \underbrace{U_{\text{ideal}}(t)}_{=\,2R_t}\;-\;\underbrace{U_{\text{actual}}^{\mathsf{Collapse}}}_{=\,2P}\;-\;C_c(t)\;=\;(2R_t-2P)-C_c(t).
\]
The adversary receives no further payments after this terminal reward. The players' subsequent payoffs remain $2P$ in total, and the adversary's continuation value is zero by the payment convention.

\paragraph{Continuation and Manipulation of the Fragile Band}
Under active maintenance let the \textbf{state variables} be $s_t:=(w_t,R_t)$, and let the adversary’s control $u_t\in\mathcal{U}$ induce the stochastic dynamics
\[
w_{t+1}=G(w_t,u_t,\xi_t),\qquad R_{t+1}=H(R_t,\eta_t),
\]
with noise terms $\xi_t,\eta_t$. The relevant band is $\mathcal B(R)=\{w:w_{\min}(R)\le w\le w_{\max}\}$, since the lower threshold depends on the cooperative payoff. To model guaranteed maintenance, restrict controls to $\mathcal U_B(s_t)$ satisfying $w_{t+1}\in\mathcal B(R_{t+1})$ almost surely. When that set is empty, maintenance is unavailable. For bounded rewards and finite control sets the controlled continuation value is
\[
 M_t=\max_{u\in\mathcal U_B(s_t)}
 \{\delta E[V_{t+1}(s_{t+1})\mid s_t,u]-C_m(t,u)\}.
\]
Here and below the continuation value uses the feasible actions at the next state. Maintaining the band preserves the availability of the two symmetric stage equilibria. During suspension the next state follows a passive transition law $K_0$ and may fall outside the band. The next decision uses the resulting state. Collapse induction remains available at its state-dependent cost. In particular, at $w>w_{\max}$ with ordered thresholds, mutual cooperation is the unique equilibrium. To induce mutual defection there the collapse action must change incentives or the subsequent game. The general Bellman equation takes the induction cost as given. The worked example below obtains it from the adjustment needed to sustain mutual defection under a specified selection rule.

\subsubsection{Conditions Determining Strategic Regimes}
\label{sec:appendix_thresholds_jp}

\noindent\textit{Assumptions.}\;
$0 < \delta < 1,\quad C_c(t)\ge 0,\quad C_m(t)\ge 0,\quad T>R>P>S.$

To compare collapse and a single maintenance action, define
\begin{equation}
    \Delta_t \coloneqq \delta E_t^m[V_{t+1}] - (2R_t - 2P),
    \qquad
    \Delta C_t \coloneqq C_m(t) - C_c(t),
\end{equation}
so that $M_t-S_t=\Delta_t-\Delta C_t$. With several feasible controls, compare $S_t$ with the maximized $M_t$ above.

\paragraph{Collapse versus maintenance}
\begin{equation}
    (2R_t - 2P) - C_c(t) \ge \delta E_t^m[V_{t+1}] - C_m(t)
    \quad\Longleftrightarrow\quad
    \Delta_t \le \Delta C_t.
\end{equation}

\paragraph{Strict preference for maintenance over collapse}
\begin{equation}
    \delta E_t^m[V_{t+1}] - (2R_t - 2P) > C_m(t) - C_c(t)
    \quad\Longleftrightarrow\quad
    \Delta_t > \Delta C_t.
\end{equation}

\paragraph{Suspension and re-intervention}
At $\Delta_t=\Delta C_t$, collapse and maintenance tie. Both are optimal if their common value is at least $A_t$. Suspension is optimal if $A_t\ge\max\{S_t,M_t\}$. When passive and maintained dynamics give the same expected continuation value, any positive maintenance cost makes suspension preferable to maintenance.

\paragraph{From incentives to the cost of collapse}
Suppose the collapse action gives a public instruction to defect. At a fixed coordination cost $K>0$, both players follow the instruction if defection is a best response to defection, including at indifference. This is the selection assumption at induction. Before the instruction the adversary may lower $w$ to $w'=w-d$ at a one-time cost $\lambda d$, with $0\le d\le w$ and $\lambda>0$. The new ratio and the selected outcome persist after collapse. With utilities divided by $a$, a player's gain from cooperating against defection is
\[
 (S+w'T)-(1+w')P=(T-P)(w'-w_{\max}).
\]
Thus the smallest feasible reduction and total induction cost are
\[
 d^*(w)=(w-w_{\max})_+,\qquad C_c(w)=K+\lambda d^*(w),
 \qquad (z)_+=\max\{z,0\}.
\]
The same fixed cost covers coordination inside and outside the band. The variable cost buys the change in incentives needed above its upper edge. At $w'=w_{\max}$ defection is a weak best response. Induction here uses the stated selection rule. Making defection strictly dominant against both actions would require a different target.

\paragraph{When to maintain and when to return}
Let $(T,P,S)=(5,2,0)$, with cooperative payoffs $R_L=4$ and $R_H=9/2$. Each period in $L$, surplus rises to $H$ with probability $p$ and otherwise stays low. Once high, it stays high. This process is the same under maintenance and suspension. Maintenance sets next period's recognition ratio to $w_F=1/2$, inside both fragile bands. Under passive evolution it moves to $w_U=3/4$, above their common upper threshold $2/3$. Against defection, cooperation earns $1/2$ less than defection at $w_F$ and $1/4$ more at $w_U$. The required reductions are therefore zero and $3/4-2/3=1/12$. Taking $K=1/5$ and $\lambda=12/5$ gives $C_c(w_F)=1/5$ and $C_c(w_U)=1/5+(12/5)(1/12)=2/5$, regardless of how the system arrived there. Maintenance costs $1/10$ and $\delta=19/20$.

In the high state, immediate collapse yields $24/5$ at $w_F$ and $23/5$ at $w_U$. Even $\delta(24/5)=4.56$, an upper bound on the payoff from waiting one period, is below both stopping values. At low surplus the best current stopping payoff is $19/5$. Consider maintaining at every low state or suspending at every low state, in either case inducing collapse on the first visit to $H$. Their respective values are
\[
 W_m=\frac{\delta p(24/5)-1/10}{1-\delta(1-p)},\qquad
 W_0=\frac{\delta p(23/5)}{1-\delta(1-p)}.
\]
If $W=\max\{W_m,W_0\}>19/5$, assigning $W$ to both low states solves both Bellman equations. Indeed, each continuation branch is affine with slope $\delta(1-p)<1$ and has its own policy value as its fixed point. At $W$ the branch with the larger fixed point returns $W$ and the other returns at most $W$. Thus in this region maintenance beats suspension exactly when $\delta p(1/5)>1/10$. The left side is the expected discounted saving on collapse costs if the high surplus arrives next period.

For $p=3/4$ the inequality is strict and $W_m=1328/305>19/5$. The adversary maintains the band as long as surplus remains low. For $p=1/3$ the inequality reverses and $W_0=437/110>19/5$. The adversary suspends while surplus is low, then re-intervenes at $w_U$ as soon as surplus rises. Both choices persist on open intervals of arrival probabilities.

\emph{Duration.} From a low state, the expected number of maintenance periods before the first high state under the maintenance policy is $1/p$. With the displayed costs and discount factor, strict maintenance throughout the low state requires $p>10/19$, giving an expected duration below $19/10$ periods. Longer optimal delays are possible at other costs and discount factors. Keeping the same payoffs, transitions conditional on arrival, and collapse costs, take $p=1/100$, $\delta=999/1000$, and $C_m=1/100000$. The Bellman values at $(H,w_F)$ and $(H,w_U)$ are respectively $24/5$ (collapse) and $\delta(24/5)-C_m=479519/100000$ (restore the band). At both low states $V_L=(\delta p(24/5)-C_m)/(1-\delta(1-p))=23971/5495>19/5$. Maintenance exceeds suspension there by $\delta p(24/5-479519/100000)-C_m>0$. From low surplus and $w_F$, this gives 100 maintenance periods in expectation and eventual collapse with probability one.

\emph{Realized welfare during maintenance.} At $(L,w_F)$, $(C,C)$ and $(D,D)$ are strict stage equilibria. The gain from cooperating against $C$ is $1$ and against $D$ is $-1/2$ after normalizing $a=1$. The symmetric mixed equilibrium therefore has cooperation probability $1/3$. Mutual cooperation gives total welfare $8$ and zero shortfall, mutual defection gives $4$ and a shortfall of $4$, and independent mixing gives expected welfare $44/9$ and a shortfall of $28/9$. The transition law and stopping rewards in this example do not depend on which stage equilibrium is played. The Bellman policy therefore determines the duration of exposure to fragility while leaving realized welfare during that exposure to equilibrium selection. A terminal shortfall relative to cooperation also need not equal the welfare change caused by the collapse action if players were already defecting.

\emph{Check within the band.} Move the passive ratio to $3/5$ and hold $K$, $\lambda$, and the surplus process fixed. Now both ratios support mutual defection and both induction costs equal $K$. The stopping rewards depend only on surplus, whose transition law is identical under maintenance and suspension. The optimal value consequently depends only on surplus: the reduced Bellman equation for stopping or suspending also solves the full equation, since maintenance delivers the same expected continuation less $C_m>0$. Paid maintenance is strictly dominated here. The cost advantage in the original example comes from avoiding an incentive adjustment above the upper edge of the band.

\subsection{Robustness to Nonlinear Generalizations}
\label{sec:appendix_nonlinear_jp}

For a nonlinear response to the recognition ratio, the equilibrium inequalities can be evaluated at the effective weight $F(w)$.

Consider the more general specification:
\begin{equation}
    u'_i = a u_i + a F(w) u_j, \quad \text{where } w = b/a.
\end{equation}
Here, $F(w)$ captures the effective influence of the \textbf{ratio of mutual recognition} $w$ and is assumed to satisfy:
\[
\begin{aligned}
&F(0)=0,\\
&F(w)\ \text{is continuous and nondecreasing},\\
&F'(w)\ge 0,\\
&F(w)\in[0,1].
\end{aligned}
\]
These assumptions make the weight assigned to the other player's payoff nondecreasing.

Under these assumptions, the equilibrium conditions generalize as follows:
\begin{itemize}
    \item $(C,C)$ is a Nash equilibrium $\iff F(w) \ge w_{\min}$,
    \item $(D,D)$ is a Nash equilibrium $\iff F(w) \le w_{\max}$.
\end{itemize}
The band in $w$ is the inverse image $F^{-1}([w_{\min},w_{\max}])$. It is nonempty exactly when the range of $F$ intersects the threshold interval. Monotonicity makes a nonempty inverse image an interval, possibly a single point. For example, with $(T,R,P,S)=(5,4,2,0)$ the thresholds are $1/4$ and $2/3$. The strictly increasing response $F(w)=w/[8(1+w)]$ for $w\ge0$ stays below $1/8$, so this game has no dual equilibrium band under that response.

\subsection{Reference-Dependent Difference-Based Payoff}
\label{sec:appendix_diffpayoff}

\paragraph{Setup.}
Let $x_t \in \mathbb{R}$ denote an observable system state (e.g., $U_{\text{actual}}(t)$), 
$\hat{x}_t$ the agent’s forecast, and $x^\ast$ a reference (normative) level.
Define three differences:
\[
\begin{aligned}
\Delta x_t &\coloneqq x_t - x_{t-1} &&(\text{change}),\\
\epsilon_t &\coloneqq x_t - \hat{x}_t &&(\text{surprise}),\\
\xi_t      &\coloneqq x_t - x^\ast   &&(\text{norm deviation}).
\end{aligned}
\]

\paragraph{General payoff (reference-dependent).}
For agent $i$,

\begin{align}
U_i(x_t;\theta_i)
&= \alpha_i\,g_1(\Delta x_t)
 + \beta_i^+\,g_2((\epsilon_t)_+)
 + \beta_i^-\,g_2((\epsilon_t)_-) \nonumber\\
&\quad
 + \gamma_i^+\,g_3((\xi_t)_+)
 + \gamma_i^-\,g_3((\xi_t)_-)
 + \delta_i\,h(x_t) - c_i.
\end{align}

where $(z)_+=\max\{z,0\},\ (z)_-=\max\{-z,0\}$, $g_k$ are continuous (locally Lipschitz), $g_k(0)=0$, and $h$ is bounded on the relevant state space.
Assume bounded stage payoffs on the reachable state space, including after a reference shift. With $0<\delta<1$ the discounted Bellman operator on bounded values is a contraction.

\paragraph{Mapping to the main adversary.}
Let $x_t=U_{\text{actual}}(t)$ and $x^\ast=U_{\text{ideal}}(t)$. 
The adversary utility $u_D=U_{\text{ideal}}-U_{\text{actual}}$ corresponds to
\[
\begin{aligned}
\alpha_i&=0,\quad \beta_i^\pm=0,\\
\gamma_i^+&=-1,\quad\gamma_i^-=1,\quad g_3(z)=z,\quad \delta_i=0
\end{aligned}
\]
with $c_i=0$, since $(\xi_t)_--(\xi_t)_+=x^\ast-x_t$. Action-dependent costs can be deducted separately. A growth-averse variant uses $\delta_i<0$ with $h(x)=x$ on a bounded domain. This identifies the gross reward at collapse. To recover the dynamic baseline, apply that reward only when collapse is induced, deduct the maintenance cost during active maintenance, and assign zero current reward during suspension. For the general stage-payoff model the state should include any lagged observation or forecast information needed to determine the reward and the next-state distribution.

\begin{proposition}[Reference-shift stability]\label{prop:ref-shift}
Assume $g_3$ is $L$-Lipschitz on the relevant domain in $|\xi|$, and the reference is shifted
by a constant $\kappa_{\mathrm{ref}}\in\mathbb{R}$ at all dates. Hold the feasible actions, forecasts, dynamics and other payoff terms fixed. Let $L_{\gamma,i}=\max\{|\gamma_i^+|,|\gamma_i^-|\}$. For each fixed policy and also after optimization, the respective value functions satisfy
\[
\bigl|V^{(\kappa_{\mathrm{ref}})}(s)-V^{(0)}(s)\bigr|
\;\le\; \frac{L_{\gamma,i}\,L\,|\kappa_{\mathrm{ref}}|}{1-\delta}.
\]
Write $B=L_{\gamma,i}L|\kappa_{\mathrm{ref}}|/(1-\delta)$. An action that leads every alternative by more than $2B$ remains the unique maximizer after the shift.
\end{proposition}

\begin{proof}[Sketch]
On each side of zero the reference term is $L_{\gamma,i}L$-Lipschitz. For a shift across zero, split the interval at zero and use $g_3(0)=0$. The absolute change in each stage reward is thus bounded by $L_{\gamma,i}L|\kappa_{\mathrm{ref}}|$. Summing the discounted changes gives $B$ uniformly over policies. Taking a supremum, including one restricted to policies with a prescribed first action, preserves this bound. Subtracting two action values can change their difference by at most $2B$.
\end{proof}

\subsection{Mass Simultaneous Intervention (Buzz / Backlash)}
\label{sec:appendix_mass_intervention}

We model many small agents acting simultaneously (e.g., like/share/praise vs. attack/criticize). Using $U_i$ from \S\ref{sec:appendix_diffpayoff}, the positive/negative sensitivity to surprise and norm deviation is separated by $(\beta_i^+,\beta_i^-)$ and $(\gamma_i^+,\gamma_i^-)$.

\paragraph{Aggregate response.}
Assume the following representative logistic response for each channel,
\[
\begin{aligned}
P_t &= \sigma\!\left(\eta[\beta^+ g_2((\epsilon_t)_+) + \gamma^+ g_3((\xi_t)_+)] - \bar c\right),\\
N_t &= \sigma\!\left(\eta[\beta^- g_2((\epsilon_t)_-) + \gamma^- g_3((\xi_t)_-)] - \bar c\right),
\end{aligned}
\]
where $P_t$ is the praise rate (buzz), $N_t$ the attack rate (backlash), $\eta>0$ sensitivity, $\bar c$ a representative cost, and $\sigma(z)=1/(1+e^{-z})$ is the logistic function.
\noindent
Here $(\beta^\pm,\gamma^\pm)$ are representative coefficients. With heterogeneous individual indices $s_i^\pm$, the exact mean rate is $E[\sigma(s_i^\pm)]$. Substituting a population mean index into the logistic function generally gives a different rate.

\paragraph{State update.}
We consider the one-line macro dynamics
\[
x_{t+1} = x_t + \kappa\,(P_t - N_t) - \rho\,(x_t-\bar x),
\]
where $\kappa>0$ is the diffusion gain, $\rho>0$ is damping, and $\bar x$ is a baseline.

\paragraph{Local resonance (buzz/backlash) condition.}
Let $\sigma(z)=1/(1+e^{-z})$, and define
\[
\begin{aligned}
s_t^+ &:= \eta[\beta^+ g_2((\epsilon_t)_+) + \gamma^+ g_3((\xi_t)_+)] - \bar c,\\
s_t^- &:= \eta[\beta^- g_2((\epsilon_t)_-) + \gamma^- g_3((\xi_t)_-)] - \bar c.
\end{aligned}
\]
Hold the forecast and reference fixed and write $f(x)$ for the state update. Define $G_t=\kappa(P-N)'(x_t)$ wherever the net response is differentiable. When $\epsilon_t\ne0$ and $\xi_t\ne0$ and the active response branches are differentiable, the chain rule gives
\[
{
\begin{aligned}
G_t = \kappa\Big(&\sigma'(s_t^+)\,\eta\,[\beta^+ g_2'((\epsilon_t)_+) \mathbf{1}_{\{\epsilon_t>0\}}
+ \gamma^+ g_3'((\xi_t)_+) \mathbf{1}_{\{\xi_t>0\}}]\\
&{}+ \sigma'(s_t^-)\,\eta\,[\beta^- g_2'((\epsilon_t)_-) \mathbf{1}_{\{\epsilon_t<0\}}
+ \gamma^- g_3'((\xi_t)_-) \mathbf{1}_{\{\xi_t<0\}}]\Big),
\end{aligned}
}
\]

where $\sigma'(z)=\sigma(z)(1-\sigma(z))$.
Each inactive summand is zero. The derivative of the negative part on its active branch cancels the minus sign in $P_t-N_t$.

At $\epsilon_t=0$ or $\xi_t=0$, compute the left and right slopes of the full update using
\[
 J_\pm(x)=\lim_{h\to0^\pm}\frac{f(x+h)-f(x)}{h}.
\]
If these limits exist and agree, their common value is $f'(x)$, even if the individual response channels have kinks. For example, take $\hat x=x^\ast=\bar x=\bar c=0$, $\beta^\pm=0$, $\gamma^\pm=\eta=1$, $g_3(z)=z$, $\kappa=4$, and $\rho=1/2$. On both sides of zero, $P-N=\sigma(x)-1/2$. Thus $f(0)=0$ and $J_-(0)=J_+(0)=1/2+4\sigma'(0)=3/2$.

\begin{proposition}[Local stability threshold]\label{prop:local-stability}
At a differentiable fixed point $x_t$ of the autonomous update, put $J_t=f'(x_t)=1+G_t-\rho$. At a switching boundary, use the common one-sided slope when it exists. The strict local tests are as follows.
\begin{enumerate}
  \item (\emph{Stability}) If $|J_t|<1$ (equivalently, $\rho-2 < G_t < \rho$), perturbations decay locally.
  \item (\emph{Buzz / monotone escape}) If $J_t>1$ (equivalently, $G_t>\rho$), a sufficiently small nonzero deviation grows in magnitude and retains its sign at each step while the trajectory remains near the fixed point.
  \item (\emph{Backlash / alternating escape}) If $J_t<-1$ (equivalently, $G_t<\rho-2$), a sufficiently small nonzero deviation grows in magnitude and reverses sign at each step while the trajectory remains near the fixed point.
\end{enumerate}
\end{proposition}

\paragraph{Unequal one-sided slopes.}
At a fixed point with unequal one-sided slopes, local contraction still follows if both slopes have absolute value strictly below one and each branch is continuously differentiable up to the fixed point. The fixed-point tests concern local responses of $x_t$. A conclusion about cooperation would additionally require a relation between $x_t$ and the recognition ratio $w$.

\clearpage
\section{Extensibility (Research Outlook)}
\label{sec:appendix_extensions_jp}

This core framework can be extended to more complex scenarios.

\subsection{Dynamics with Multiple Adversaries}
Consider a set of competing adversaries $\mathcal{D} = \{D_1, \dots, D_n\}$. Model the evolution of system fragility $w_t$ as a function of the aggregate of their interventions. The problem then becomes a game among adversaries. Coordination among adversaries may, paradoxically, stabilize fragility, whereas rivalry can precipitate \emph{premature and suboptimal collapse}.

\subsection{Networked Fragility}
Let a network $G=(V,E)$ represent subsystems (nodes) each with a local fragility $w_i(t)$. Adversarial influence can propagate along network links. In this setting, overall system stability depends on the spectral properties of the network’s adjacency matrix, providing a bridge to models of \emph{systemic risk} and \emph{information cascades}.

\clearpage
\bibliographystyle{plain}
\bibliography{references}

@article{Barreno2010,
  title = {The security of machine learning},
  volume = {81},
  ISSN = {1573-0565},
  url = {http://dx.doi.org/10.1007/s10994-010-5188-5},
  DOI = {10.1007/s10994-010-5188-5},
  number = {2},
  journal = {Machine Learning},
  publisher = {Springer Science and Business Media LLC},
  author = {Barreno,  Marco and Nelson,  Blaine and Joseph,  Anthony D. and Tygar,  J. D.},
  year = {2010},
  month = may,
  pages = {121–148}
}

@inbook{Biggio2013,
  title = {Evasion Attacks against Machine Learning at Test Time},
  ISBN = {9783642387098},
  ISSN = {1611-3349},
  url = {http://dx.doi.org/10.1007/978-3-642-40994-3_25},
  DOI = {10.1007/978-3-642-40994-3_25},
  booktitle = {Advanced Information Systems Engineering},
  publisher = {Springer Berlin Heidelberg},
  author = {Biggio,  Battista and Corona,  Igino and Maiorca,  Davide and Nelson,  Blaine and Šrndić,  Nedim and Laskov,  Pavel and Giacinto,  Giorgio and Roli,  Fabio},
  year = {2013},
  pages = {387–402}
}

@inproceedings{Dalvi2004,
  series = {KDD04},
  title = {Adversarial classification},
  url = {http://dx.doi.org/10.1145/1014052.1014066},
  DOI = {10.1145/1014052.1014066},
  booktitle = {Proceedings of the tenth ACM SIGKDD international conference on Knowledge discovery and data mining},
  publisher = {ACM},
  author = {Dalvi,  Nilesh and Domingos,  Pedro and Mausam and Sanghai,  Sumit and Verma,  Deepak},
  year = {2004},
  month = aug,
  pages = {99–108},
  collection = {KDD04}
}

@article{Dawes1980,
  title = {Social Dilemmas},
  volume = {31},
  ISSN = {1545-2085},
  url = {http://dx.doi.org/10.1146/annurev.ps.31.020180.001125},
  DOI = {10.1146/annurev.ps.31.020180.001125},
  number = {1},
  journal = {Annual Review of Psychology},
  publisher = {Annual Reviews},
  author = {Dawes,  R M},
  year = {1980},
  month = jan,
  pages = {169–193}
}

@article{Hardin1968,
  title = {The Tragedy of the Commons: The population problem has no technical solution; it requires a fundamental extension in morality.},
  volume = {162},
  ISSN = {1095-9203},
  url = {http://dx.doi.org/10.1126/science.162.3859.1243},
  DOI = {10.1126/science.162.3859.1243},
  number = {3859},
  journal = {Science},
  publisher = {American Association for the Advancement of Science (AAAS)},
  author = {Hardin,  Garrett},
  year = {1968},
  month = dec,
  pages = {1243–1248}
}

@article{Korzhyk2011,
  author = {Korzhyk, Dmytro and Yin, Zhengyu and Kiekintveld, Christopher and Conitzer, Vincent and Tambe, Milind},
  title = {{Stackelberg} vs. {Nash} in Security Games: An Extended Investigation of Interchangeability, Equivalence, and Uniqueness},
  journal = {Journal of Artificial Intelligence Research},
  volume = {41},
  pages = {297--327},
  year = {2011},
  doi = {10.1613/jair.3269}
}

@article{Lamport1982,
  title = {The Byzantine Generals Problem},
  volume = {4},
  ISSN = {1558-4593},
  url = {http://dx.doi.org/10.1145/357172.357176},
  DOI = {10.1145/357172.357176},
  number = {3},
  journal = {ACM Transactions on Programming Languages and Systems},
  publisher = {Association for Computing Machinery (ACM)},
  author = {Lamport,  Leslie and Shostak,  Robert and Pease,  Marshall},
  year = {1982},
  month = jul,
  pages = {382–401}
}

@article{Levine1998,
  title = {Modeling Altruism and Spitefulness in Experiments},
  volume = {1},
  ISSN = {1094-2025},
  url = {http://dx.doi.org/10.1006/redy.1998.0023},
  DOI = {10.1006/redy.1998.0023},
  number = {3},
  journal = {Review of Economic Dynamics},
  publisher = {Elsevier BV},
  author = {Levine,  David K.},
  year = {1998},
  month = jul,
  pages = {593–622}
}

@article{Loomes1982,
  title = {Regret Theory: An Alternative Theory of Rational Choice Under Uncertainty},
  volume = {92},
  ISSN = {0013-0133},
  url = {http://dx.doi.org/10.2307/2232669},
  DOI = {10.2307/2232669},
  number = {368},
  journal = {The Economic Journal},
  publisher = {Oxford University Press (OUP)},
  author = {Loomes,  Graham and Sugden,  Robert},
  year = {1982},
  month = dec,
  pages = {805}
}

@inproceedings{Lowd2005,
  series = {KDD05},
  title = {Adversarial learning},
  url = {http://dx.doi.org/10.1145/1081870.1081950},
  DOI = {10.1145/1081870.1081950},
  booktitle = {Proceedings of the eleventh ACM SIGKDD international conference on Knowledge discovery in data mining},
  publisher = {ACM},
  author = {Lowd,  Daniel and Meek,  Christopher},
  year = {2005},
  month = aug,
  pages = {641–647},
  collection = {KDD05}
}

@inproceedings{paruchuri2008playing,
  title={Playing games for security: An efficient exact algorithm for solving Bayesian Stackelberg games},
  author={Paruchuri, Praveen and Pearce, Jonathan P and Marecki, Janusz and Tambe, Milind and Ordonez, Fernando and Kraus, Sarit},
  booktitle={Proceedings of the 7th international joint conference on Autonomous agents and multiagent systems-Volume 2},
  pages={895--902},
  year={2008}
}

@book{tambe2011security,
  title={Security and game theory: algorithms, deployed systems, lessons learned},
  author={Tambe, Milind},
  year={2011},
  publisher={Cambridge university press}
}

@article{Fehr1999,
  title = {A Theory of Fairness,  Competition,  and Cooperation},
  volume = {114},
  ISSN = {1531-4650},
  url = {http://dx.doi.org/10.1162/003355399556151},
  DOI = {10.1162/003355399556151},
  number = {3},
  journal = {The Quarterly Journal of Economics},
  publisher = {Oxford University Press (OUP)},
  author = {Fehr,  E. and Schmidt,  K. M.},
  year = {1999},
  month = aug,
  pages = {817–868}
}

@article{Bolton2000,
  title = {ERC: A Theory of Equity,  Reciprocity,  and Competition},
  volume = {90},
  ISSN = {0002-8282},
  url = {http://dx.doi.org/10.1257/aer.90.1.166},
  DOI = {10.1257/aer.90.1.166},
  number = {1},
  journal = {American Economic Review},
  publisher = {American Economic Association},
  author = {Bolton,  Gary E and Ockenfels,  Axel},
  year = {2000},
  month = mar,
  pages = {166–193}
}

@article{Acemoglu2000,
  title = {Political Losers as a Barrier to Economic Development},
  volume = {90},
  ISSN = {0002-8282},
  url = {http://dx.doi.org/10.1257/aer.90.2.126},
  DOI = {10.1257/aer.90.2.126},
  number = {2},
  journal = {American Economic Review},
  publisher = {American Economic Association},
  author = {Acemoglu,  Daron and Robinson,  James A},
  year = {2000},
  month = may,
  pages = {126–130}
}

@article{Banerjee1993,
  title = {Occupational Choice and the Process of Development},
  volume = {101},
  ISSN = {1537-534X},
  url = {http://dx.doi.org/10.1086/261876},
  DOI = {10.1086/261876},
  number = {2},
  journal = {Journal of Political Economy},
  publisher = {University of Chicago Press},
  author = {Banerjee,  Abhijit V. and Newman,  Andrew F.},
  year = {1993},
  month = apr,
  pages = {274–298}
}

@article{dasgupta1986inequality,
  title={Inequality as a determinant of malnutrition and unemployment: Theory},
  author={Dasgupta, Partha and Ray, Debraj},
  journal={The Economic Journal},
  volume={96},
  number={384},
  pages={1011--1034},
  year={1986},
  publisher={Oxford University Press Oxford, UK}
}

@article{abbink2009pleasure,
  title={The pleasure of being nasty},
  author={Abbink, Klaus and Sadrieh, Abdolkarim},
  journal={Economics letters},
  volume={105},
  number={3},
  pages={306--308},
  year={2009},
  publisher={Elsevier}
}

@article{cooper1988coordinating,
  title={Coordinating coordination failures in Keynesian models},
  author={Cooper, Russell and John, Andrew},
  journal={The Quarterly Journal of Economics},
  volume={103},
  number={3},
  pages={441--463},
  year={1988},
  publisher={MIT Press}
}

@article{azariadis1990threshold,
  title={Threshold externalities in economic development},
  author={Azariadis, Costas and Drazen, Allan},
  journal={The quarterly journal of economics},
  volume={105},
  number={2},
  pages={501--526},
  year={1990},
  publisher={MIT Press}
}

@article{murphy1989industrialization,
  title={Industrialization and the big push},
  author={Murphy, Kevin M and Shleifer, Andrei and Vishny, Robert W},
  journal={Journal of political economy},
  volume={97},
  number={5},
  pages={1003--1026},
  year={1989},
  publisher={The University of Chicago Press}
}
\end{document}